\documentclass{llncs}

\usepackage{amssymb,amsmath,graphics,graphicx,color}

\newtheorem{observation}{Observation}

\newenvironment{my_enumerate}{
\begin{enumerate}
  \setlength{\itemsep}{.8pt}
  \setlength{\parskip}{0pt}
  \setlength{\parsep}{0pt}}{\end{enumerate}
}

\begin{document}
\pagestyle{plain}

\title{Obtaining a Bipartite Graph by Contracting\\ Few Edges\thanks{This work is supported by the Research Council of Norway.}}

\author{Pinar Heggernes$^1$ \and Pim van 't Hof$^1$ \and Daniel Lokshtanov$^2$ \and Christophe Paul$^3$}

\institute{
$^1$ Department of Informatics, University of Bergen, Norway.\\
\texttt{\{pinar.heggernes,pim.vanthof\}@ii.uib.no}\\
$^2$ Dept. Computer Science and Engineering, University of California San Diego, USA.\\
\texttt{dlokshtanov@cs.ucsd.edu}\\
$^3$ CNRS, LIRMM, Universit{\'e} Montpellier 2, France.
\texttt{paul@lirmm.fr}
}

\date{}
\maketitle

\begin{abstract}
We initiate the study of the {\sc Bipartite Contraction} problem from the perspective of parameterized complexity. In this problem we are given a graph $G$ and an integer $k$, and the task is to determine whether we can obtain a bipartite graph from $G$ by a sequence of at most $k$ edge contractions. Our main result is an $f(k)\, n^{O(1)}$ time algorithm for {\sc Bipartite Contraction}. Despite a strong resemblance between {\sc Bipartite Contraction} and the classical {\sc Odd Cycle Transversal} (OCT) problem, the methods developed to tackle OCT do not seem to be directly applicable to {\sc Bipartite Contraction}. Our algorithm is based on a novel combination of the {\em irrelevant vertex} technique, introduced by Robertson and Seymour, and the concept of {\em important separators}. Both techniques have previously been used as key components of algorithms for fundamental problems in parameterized complexity. However, to the best of our knowledge, this is the first time the two techniques are applied in 
unison.
\end{abstract}

\section{Introduction}

{\sc Odd Cycle Transversal (OCT)} is a central problem in parameterized complexity. The establishment of its fixed parameter tractability by Reed, Smith, and Vetta~\cite{RSV04} in 2004, settling a long-standing open question~\cite{DF99}, supplied the field with the powerful new technique of iterative compression~\cite{NBook}. {\sc OCT} and the closely related {\sc Edge Bipartization} problem take a graph $G$ and an integer $k$ as input, and ask whether a bipartite graph can be obtained by deleting at most $k$ vertices, respectively $k$ edges, from $G$. These two problems can be viewed as two ways of measuring how close $G$ is to being bipartite. Over the last few years a considerable amount of research has been devoted to studying different measures of how close a graph is to being bipartite~\cite{GGHNW06,H09,KR-STOC10,KR-SODA10}, and how similarity to a bipartite graph can be exploited~\cite{DemaineHK10}.

A natural similarity measure is defined by the {\sc Bipartite Contraction} problem: Given a graph $G$ and an integer $k$, can we obtain a bipartite graph from $G$ by a sequence of at most $k$ edge contractions in $G$? The number of possible edge contractions in $G$ is always less than the number of vertices of $G$, since every edge contraction reduces the number of vertices by exactly one. In practical instances of bipartization problems the similarity parameter $k$ tends to be small~\cite{RBIL02}, which makes these problems especially well-suited for parameterized algorithms. A graph problem with input $G$ and $k$ is {\it fixed parameter tractable (FPT)} if there is an algorithm with running time $f(k) \, n^{O(1)}$, where the function $f$ depends only on $k$ and not on the size of $G$. Considering the significant amount of interest the problems OCT and {\sc Edge Bipartization} have received, we find it surprising that {\sc Bipartite Contraction} has not yet been studied.

In this paper we show that {\sc Bipartite Contraction} is fixed parameter tract\-able when parameterized by the number $k$ of edges to be contracted. The key ingredients of our algorithm fundamentally differ from the ones used in the above-mentioned algorithms for OCT and {\sc Edge Bipartization}. In the algorithm for {\sc OCT} by Reed, Smith, and Vetta \cite{RSV04}, iterative compression is combined with maximum flow arguments. The recent nearly linear time algorithm for the two problems, due to Kawarabayashi and Reed \cite{KR-SODA10}, uses the notion of odd minors, together with deep structural results of Robertson and Seymour~\cite{RS-GMXIII} about graphs of large treewidth without large clique minors. Interestingly, {\sc Bipartite Contraction} does not seem to be amenable to these approaches.

Although our algorithm is based on iterative compression, it seems difficult to adapt the compression step from~\cite{RSV04} for OCT to work for {\sc Bipartite Contraction}. Instead, we perform the compression step using a variant of the {\em irrelevant vertex technique}, introduced by Robertson and Seymour~\cite{RS-GMXIII} (see also \cite{RS-GMXXII}). In particular, if the treewidth of the input graph is large, then we identify an irrelevant edge that can be deleted from the graph without affecting the outcome. The irrelevant vertex technique has played a key role in the solutions of several problems (see e.g.,~\cite{K-IPCO08,KR-STOC10,KW-STOC10}).

Our algorithm crucially deviates from previous work in the manner in which it finds the irrelevant edge. While previous work has relied on large minor models as obstructions to small treewidth, ours uses the fact that any graph of high treewidth contains a large $p$-{\em connected} set $X$~\cite{DGJT99}. A vertex set $X$ is $p$-connected if, for any two subsets $X_1$ and $X_2$ of $X$ with $|X_1|=|X_2|\leq p$, there are $|X_1|$ vertex-disjoint paths with one endpoint in $X_1$ and the other in $X_2$. Using $p$-connected sets in order to find irrelevant edges has several advantages. First, our algorithm avoids the huge parameter-dependence which seems to be an inadvertent side effect of applying the Robertson and Seymour's graph minors machinery. Second, our arguments are nearly self-contained, and rely only on results whose proofs are simple enough to be taught in a graduate class. Thus, even though our algorithm is not practical by any means, it is much closer to practicality than other algorithms based on the irrelevant vertex technique. It is an intriguing question whether some of the algorithms that currently use Robertson-Seymour machinery to find an irrelevant vertex can be modified in such a way, that they find an irrelevant vertex using $p$-connected sets instead.

Using $p$-connected sets in order to find an irrelevant vertex or edge is non-trivial, because $p$-connectivity is a more ``implicit'' notion than that of a large minor model. We overcome this difficulty by using {\em important sets}. Important sets and the closely related notion of {\em important separators} were introduced in \cite{marx-separation-full} to prove the fixed-parameter tractability of multiway cut problems. The basic idea is that in many problems where terminals need to be separated in some way, it is sufficient to consider separators that are ``as far as possible'' from one of the terminals. Important separators turned out to be a crucial component, in some cases implicitly, in the solutions of cardinal problems in parameterized complexity~\cite{chenll07,dblp:journals/jacm/chenllor08,DBLP:conf/icalp/RazgonO08,MR11}. To the best of our knowledge, this is the first time the irrelevant vertex technique and important sets (or separators) are used together. We believe that this combination will turn out to be a useful and powerful tool.

\section{Definitions and notation}

All graphs considered in this paper are finite, undirected, and simple, i.e., do not contain multiple edges or loops. Given a graph $G$, we denote its vertex set by $V(G)$ and its edge set by $E(G)$. We also use the ordered pair $(V(G),E(G))$ to represent $G$. We let $n= |V(G)|$ and $m=|E(G)|$. For two graphs $G_1=(V_1,E_1)$ and $G_2=(V_2,E_2)$, the {\em disjoint union} of $G_1$ and $G_2$ is the graph $G_1\cup G_2=(V_1\cup V_2,E_1\cup E_2)$. The {\em deletion} of an edge $e\in E(G)$ yields the graph $G-e=(V(G),E(G)\setminus e)$. For a set $X\subseteq V(G)$, we write $G[X]$ to denote the subgraph of $G$ {\em induced} by $X$. A graph is {\it connected} if there is a path between every pair of its vertices. The {\it connected components} of a graph are its maximal connected subgraphs. For any set $X\subseteq V(G)$, we write $\delta_G(X)$ to denote the set of edges in $G$ that have exactly one endpoint in $X$. We define $d_G(X)=|\delta(X)|$.

The \emph{contraction} of edge $xy$ in $G$ deletes vertices $x$ and $y$ from $G$, and replaces them by a new vertex, which is made adjacent to precisely those vertices that were adjacent to at least one of the vertices $x$ and $y$. The resulting graph is denoted $G/xy$. Every edge contraction reduces the number of vertices in the graph by exactly one. We point out that several edges might disappear as the result of a single edge contraction. For a set $S\subseteq E(G)$, we write $G/S$ to denote the graph obtained from $G$ by repeatedly contracting an edge from $S$ until no such edge remains. It follows from the definition of an edge contraction that in order to obtain the graph $G/S$ from $S$, it is necessary and sufficient to contract all the edges of some spanning forest of the graph $(V(G),S)$.

Let $H$ be a graph with $V(H)= \{h_1, h_2, \ldots, h_\ell\}$. A graph $G$ is $H$-{\em contractible} if $H$ can be obtained from $G$ by contracting edges. Saying that $G$ is $H$-contractible is equivalent to saying that $G$ has a so-called $H$-{\it witness structure} ${\cal W}$, which is a partition of $V(G)$ into {\it witness sets} $W(h_1), W(h_2), \ldots, W(h_\ell)$. The witness sets have the property that each of them induces a connected subgraph of $G$, and for every two $h_i,h_j\in V(H)$, there is an edge in $G$ between a vertex of $W(h_i)$ and a vertex of $W(h_j)$ if and only if $h_i$ and $h_j$ are adjacent in $H$. Let $G'=G[W(h_1)]\cup \cdots \cup G[W(h_\ell)]$ be the graph obtained from $G$ by removing all the edges of $G$, apart from the ones that have both endpoints in the same witness set. In order to contract $G$ to $H$, it is necessary and sufficient to contract all the edges of some spanning forest $F$ of $G'$. Note that $|E(F)|=\sum_{i=1}^{\ell} (|W(h_i)|-1)=|V(G)|-|V(H)|$.

A $2$-{\em coloring} of a graph $G$ is a function $\phi : V(G) \rightarrow \{1,2\}$. We point out that a 2-coloring of $G$ is merely an assignment of colors 1 and 2 to the vertices of $G$, and should therefore not be confused with a {\em proper} 2-coloring of $G$, which is a 2-coloring with the additional property that no two adjacent vertices receive the same color. An edge $uv$ is said to be {\em good} (with respect to $\phi$) if $\phi(u) \neq \phi(v)$, and $uv$ is called {\em bad} (with respect to $\phi$) otherwise. A {\em good component} of $\phi$ is the vertex set of a connected component of the graph $(V(G),E')$, where $E'\subseteq E$ is the set of all edges that are good with respect to $\phi$. Any 2-coloring $\phi$ of $G$ defines a partition of $V(G)$ into two sets $V_\phi^1$ and $V_\phi^2$, which are the sets of vertices of $G$ colored 1 and 2 by $\phi$, respectively. A set $X\subseteq V(G)$ is a {\em monochromatic component} of $\phi$ if $G[X]$ is a connected component of $G[V_\phi^1]$ or a connected component of $G[V_\phi^2]$. We write ${\cal M}_\phi$ to denote the set of all monochromatic components of $\phi$. The {\em cost} of a 2-coloring $\phi$ is defined as $\sum_{X \in {\cal M}_\phi} (|X|-1)$. Note that the cost of a 2-coloring $\phi$ of $G$ is 0 if and only if $\phi$ is a proper 2-coloring of $G$. 

Let $G$ be a graph. A {\em tree decomposition} of $G$ is a pair $(T, \mathcal{ X}=\{X_{t}\}_{t\in V(T)})$, where $T$ is a tree and $\mathcal{X}$ is a collection of subsets of $V(G)$, satisfying the following three properties: (1) $\cup_{u\in V(T)}X_u=V$; (2) $\forall {uv\in E(G)}, \exists {t\in V(T)} : \{u,v\}\subseteq X_{t}$; and (3) $\forall {v\in V(G)}, \ T[\{t\mid v\in X_{t}\}]$ is connected. The {\em width} of a tree decomposition is $\max_{t\in V(T)} |X_t|-1$ and the {\em treewidth} of $G$, denoted $tw(G)$,
is the minimum width over all tree decompositions of $G$. 
 
The syntax of {\em Monadic Second Order Logic} (MSO) of graphs includes the logical connectives $\vee$, $\land$, $\neg$, variables for vertices, edges, sets of vertices and sets of edges, the quantifiers $\forall$, $\exists$ that can be applied to these variables, and the following five binary relations:
\begin{my_enumerate}
\item $u\in U$, where $u$ is a vertex variable and $U$ is a vertex set variable; 
\item $d \in D$, where $d$ is an edge variable and $D$ is an edge set variable;
\item $\mathbf{inc}(d,u)$, where $d$ is an edge variable, $u$ is a vertex variable, and the interpretation is that the edge $d$ is incident to the vertex $u$; 
\item $\mathbf{adj}(u,v)$, where  $u$ and $v$ are vertex variables and the interpretation is that $u$ and $v$ are adjacent;
\item equality of variables representing vertices, edges, sets of vertices and sets of edges.
\end{my_enumerate}

\section{{\sc Bipartite Contraction} and the cost of $2$-colorings}

In the {\sc Bipartite Contraction} problem we are given a graph $G$ and an integer $k$, and the task is to determine whether there exists a set $S \subseteq E(G)$ of at most $k$ edges such that $G/S$ is bipartite. The following lemma allows us to reformulate this problem in terms of 2-colorings of the graph $G$.

\begin{lemma}
\label{lem:contrcolor}
A graph $G$ has a 2-coloring $\phi$ of cost at most $k$ if and only if there exists a set $S\subseteq E(G)$ of at most $k$ edges such that $G/S$ is bipartite.
\end{lemma}

\begin{proof}
Suppose $G$ has a 2-coloring $\phi$ of cost at most $k$. We build an edge set $S$ as follows. For each monochromatic component $X \in {\cal M}_\phi$, find a spanning tree $T$ of $G[X]$ and add the $|X|-1$ edges of $T$ to $S$. The total number of edges in $S$ is exactly the cost of $\phi$, so $|S| \leq k$. It remains to argue that $G' = G/S$ is bipartite. Let $uv\in S$ be a bad edge in $G$, and let $\phi'$ be the 2-coloring of $G/uv$ that assigns color $\phi(u)=\phi(v)$ to the new vertex resulting from the contraction of $uv$, and that assigns color $\phi(w)$ to every vertex $w\in V(G)\setminus \{u,v\}$. Since we contracted a bad edge, the cost of $\phi'$ is 1 less than the cost of $\phi$. Repeating this for every edge in $S$ yields a 2-coloring $\phi''$ of $G'$ of cost $0$. This means that $\phi''$ is a proper 2-coloring  of $G'$, which implies that $G'$ is bipartite.

For the reverse direction, suppose there is a set $S\subseteq E(G)$ of at most $k$ edges such that $G' = G/S$ is bipartite. We define $G^*$ to be the graph with the same vertex set as $G$ and edge set $S$, i.e., $G^*=(V(G),S)$. Let ${\cal W}$ be a $G'$-witness structure of $G$. Observe that $W(y)$ induces a connected component of $G^*$ for every $y \in V(G')$. Let $\phi$ be a proper 2-coloring of $G'$. We construct a $2$-coloring $\phi'$ of $G$ as follows. For every $v \in V(G)$, we set $\phi'(v) = \phi(y)$, where $y$ is the vertex in $V(G')$ such that $v \in W(y)$. Since the monochromatic components of $\phi'$ are exactly the connected components of the graph $G^*$, and since $G^*$ contains exactly $S$ edges, the cost of $\phi'$ is at most $|S|\leq k$.
\qed
\end{proof}

An instance of the {\sc Cheap Coloring} problem consists of a graph $G$ and an integer $k$, and the task is to decide whether $G$ has a 2-coloring of cost at most $k$. Lemma~\ref{lem:contrcolor} shows that the problems {\sc Bipartite Contraction} and {\sc Cheap Coloring} are equivalent. 

The deletion of an edge can not increase the cost of a $2$-coloring, and can only decrease the cost of a 2-coloring by at most one. We state this as the following observation.

\begin{observation}
\label{obs:remedge}
Let $\phi$ be a 2-coloring of $G$ of cost $k$. For any edge $uv \in E(G)$, the cost of $\phi$ in $G-uv$ is $k$ or $k-1$. 
\end{observation} 

Observation~\ref{obs:remedge} allows us to use the well-known {\em iterative compression} technique of Reed, Smith and Vetta~\cite{RSV04} to reduce the {\sc Cheap Coloring} problem to the {\sc Cheaper Coloring} problem. The {\sc Cheaper Coloring} problem takes as input a graph $G$, an integer $k$, and a $2$-coloring $\phi$ of $G$ of cost $k+1$, and the task is to either find a $2$-coloring of $G$ of cost at most $k$, or to conclude that such a coloring does not exist.

\begin{lemma}
\label{lem:cheaptocheaper}
If there is an algorithm for {\sc Cheaper Coloring} that runs in time $f(k) \, n^c$, then there is an algorithm for {\sc Cheap Coloring} that runs in time $f(k) \, n^cm$.
\end{lemma}

\begin{proof}
Suppose there exists an algorithm for {\sc Cheaper Coloring} that runs in time $f(k) \, n^c$. Then we can solve an instance $(G,k)$ of {\sc Cheap Coloring} by iterating over the edges $e_1, e_2, \ldots e_m$ of $G$ as follows. For every $i\in \{1,\ldots,m\}$, we define $G_i$ to be the graph with vertex set $V(G)$ and edge set $E_i = \{e_j~:~j \leq i\}$. The graph $G_1$ has a 2-coloring $\phi_1$ of cost $0$, which is at most $k$. For the first $k$ iterations, we trivially maintain a 2-coloring of cost at most $k$. Now, in iteration $i$ of the algorithm, assume that we have a 2-coloring $\phi_i$ of cost at most $k$ in $G_i$. By Observation~\ref{obs:remedge}, the cost of $\phi_i$ in $G_{i+1}$ is at most $k+1$. If the cost of $\phi_i$ in $G_{i+1}$ is at most $k$, then we proceed to the $(i+1)$th iteration. Otherwise, we run the algorithm for {\sc Cheaper Coloring} with input $(G_{i+1},k,\phi_i)$. If the algorithm concludes that $G_{i+1}$ has no 2-coloring of cost at most $k$, then, by Observation~\ref{obs:remedge}, neither does $G$. If, on the other hand, the algorithm outputs a 2-coloring $\phi_{i+1}$ of $G_{i+1}$ of cost at most $k$, then we proceed to the $(i+1$)th iteration. Since we call the algorithm for {\sc Cheaper Coloring} at most $m$ times, each time with parameter $k$, the time bound follows.
\qed
\end{proof}

We have now almost reached the variant of the problem that will be the focus of attention in the remainder of this paper. For two disjoint vertex sets $T_1$ and $T_2$, a 2-coloring $\phi$ of $G$ is a $(T_1, T_2)$-{\em extension} if $\phi$ colors every vertex in $T_1$ with $1$ and every vertex in $T_2$ with $2$. In the {\sc Cheap Coloring Extension} problem we are given a {\em bipartite} graph $G$, two integers $k$ and $t$, and two disjoint vertex sets $T_1$ and $T_2$ such that $|T_1|+|T_2| \leq t$. The objective is to find a $(T_1, T_2)$-extension $\phi$ of cost at most $k$, or to conclude that such a 2-coloring does not exist. We will say that a $(T_1,T_2)$-extension $\phi$ is a {\em cheapest} $(T_1,T_2)$-extension if there is no $(T_1, T_2)$-extension $\phi'$ with strictly lower cost than $\phi$.

\begin{lemma}
\label{lem:redextension}
If there is an algorithm for {\sc Cheap Coloring Extension} that runs in time $f(k,t) \, n^c$, then there is an algorithm for {\sc Cheaper Coloring} that runs in time $4^{k+1} f(k,2k+2) \, n^c$. 
\end{lemma}

\begin{proof}
Given an $f(k,t) \, n^c$ time algorithm for {\sc Cheap Coloring Extension}, we show how to solve an instance $(G,k,\phi)$ of {\sc Cheaper Coloring}. Let $S$ be the set of all bad edges in $G$ with respect to $\phi$, and let $X$ be the set of endpoints of the edges in $S$. Since $\phi$ has cost $k+1$, we have $|X| \leq 2k+2$. We create $4^{k+1}$ instances of {\sc Cheap Coloring Extension} as follows.

For every possible partition $X$ into two sets $X_1$ and $X_2$, we set $k'=k$ and $t=|X_1|+|X_2|$, and we build a graph $G(X_1,X_2)$ from $G$ in the following way. As long as there is an edge $uv \in S$ such that $u$ and $v$ are both in $X_1$ or both in $X_2$, contract the edge $uv$, put the new vertex resulting from the contraction into the set $X_i$ that $u$ and $v$ belonged to, and decrease $k'$ by 1. Since the cost of $\phi$ is at most $k+1$, we contract at most $k+1$ edges in this way, and hence $k' \ge -1$. When there are no such edges left, then we discard this partition of $X$ into $X_1$ and $X_2$ if $k'=-1$; otherwise, we continue to build an instance of {\sc Cheap Coloring Extension} as follows. Delete all edges $uv \in S$ with $u \in X_i$ and $v \in X_j$ such that $i\neq j$. Since $S$ contains all the edges of $G$ that are bad with respect to $\phi$, and each of the edges of $S$ is either contracted or deleted, the resulting graph $G(X_1,X_2)$ has no bad edges with respect to $\phi$ and is therefore bipartite. Thus we obtain an instance $(G(X_1,X_2),k',t,X_1,X_2)$ of {\sc Cheap Coloring Extension} with $k' \ge 0$.

We now show that $(G,k,\phi)$ is a yes-instance of {\sc Cheaper Coloring} if and only if there is a partition of $X$ into $X_1$ and $X_2$ such that $(G(X_1, X_2),k',t,X_1,X_2)$ with $k'\ge 0$ is a yes-instance of {\sc Cheap Coloring Extension}.

Suppose that $(G,k,\phi)$ is a yes-instance of {\sc Cheaper Coloring}. Then there exists a 2-coloring $\phi^*$ of $G$ of cost at most $k$. Let $X_1$ and $X_2$ be the vertices of $X$ that are colored 1 and 2 by $\phi^*$, respectively. Consider the set $S'\subseteq S$ of edges that were contracted in order to obtain $G(X_1,X_2)$ from $G$ in the way described earlier. Since every edge in $S'$ is bad with respect to $\phi^*$, the cost of $\phi^*$ decreased by 1 with every edge contraction. Hence, $\phi^*$ is an $(X_1,X_2)$-extension of $G(X_1,X_2)$ of cost $k'$. We conclude that $(G(X_1,X_2),k',t,X_1,X_2)$ is a yes-instance of {\sc Cheap Coloring Extension}.

For the reverse direction, suppose there is a partition of $X$ into $X_1$ and $X_2$ such that $(G(X_1,X_2),k',t,X_1,X_2)$ is a yes-instance of {\sc Cheap Coloring Extension} with $k'\geq 0$, i.e., the bipartite graph $G(X_1,X_2)$ has an $(X_1,X_2)$-extension $\psi$ of cost at most $k'$. Let $S'\subseteq S$ be the set of edges that were contracted in $G$ to create the instance $(G(X_1,X_2),k',t,X_1,X_2)$. Since $k'=k-|S'| \ge 0$, we have that $|S'| \le k$. We define a 2-coloring $\theta$ of $G$ by coloring both endpoints of every edge $uv$ in $S'$ with the color that $\psi$ assigned to the vertex resulting from the contraction of the edge $uv$, and coloring all other vertices in $G$ with the color they received from $\psi$. Clearly, the cost of $\theta$ is at most $k' +|S'| =k$, and therefore $(G,k,\phi)$ is a yes-instance of {\sc Cheaper Coloring}.

Since we need to run the $f(k,t) \, n^c$ time algorithm for {\sc Cheap Coloring Extension} at most $4^{k+1}$ times, with parameters $k'\leq k$ and $t = |X| \le 2k+2$ at each iteration, the time bound follows.
\qed
\end{proof}

The next section is devoted to showing that {\sc Cheap Coloring Extension} is fixed parameter tractable when parameterized by $k$ and $t$. The reason we want to work with the {\sc Cheap Coloring Extension} problem rather than with the {\sc Bipartite Contraction} problem directly is that, as we shall see in Section~\ref{sec:important}, {\sc Cheap Coloring Extension} is a ``cut'' problem, and is therefore amenable to techniques based on {\em important separators} \cite{marx-separation-full}.

\section{Solving {\sc Cheap Coloring Extension} in FPT time}

In this section, we present an algorithm for the {\sc Cheap Coloring Extension} problem. For the remainder of this section, let $(G,k,t,T_1,T_2)$ be a given instance of {\sc Cheap Coloring Extension}. Recall that $G$ is bipartite. The high level structure of our algorithm is as follows. If the treewidth of $G$ is bounded by a function of $k$ and $t$, then we can use standard dynamic programming techniques to solve the problem in time $f(k,t) \, n$. If, on the other hand, the treewidth of $G$ is large, then we can find a large set which is ``highly connected''. In this case we show how to find in $f(k,t) \, n^{O(1)}$ time an edge $e \in E(G)$ such that $G$ has a $(T_1,T_2)$-extension of cost at most $k$ if and only if $G - e$ does. We then re-run our algorithm on $G - e$. 

To make the distinction between the two cases in our algorithm more precise, we use the following notion, due to Diestel et al.~\cite{DGJT99}. A set $X\subseteq V(G)$ is $p$-{\em connected} in $G$ if $|X|\geq p$ and, for all subsets $X_1,X_2\subseteq X$ with $|X_1|=|X_2|\leq p$, there are $|X_1|$ vertex-disjoint paths in $G$ with one endpoint in $X_1$ and the other in $X_2$. Diestel et al.~\cite{DGJT99} prove the following statement in the proof of Proposition 3 (ii): if $h\geq p$ and $G$ contains no $p$-connected set of size $h$, then $G$ has treewidth $< h+p-1$. (In fact, they prove a stronger version of this statement using the notion of an {\em externally} $p$-connected set, but we do not need this stronger assertion for our purposes.) We define a set $X$ to be {\em well-connected} if it is $|X|/2$-connected. Using this definition, the result of Diestel et al.~ can be seen to imply the following.

\begin{theorem}[\cite{DGJT99}]
\label{thm:twlinked}
If $tw(G) > w$, then $G$ contains a well-connected set of size at least $2w/3$.
\end{theorem}

The proof of Theorem \ref{thm:twlinked} is constructive. In fact, given $G$ and $w$, a tree decomposition of width at most $w$ or a well-connected set of size at least $2w/3$ can be computed in time $c^wn^{O(1)}$ for some constant $c$~\cite{DGJT99}. We use Theorem~\ref{thm:twlinked} to compute either a tree-decomposition of $G$ of width at most $3(4k^2)\, t \, 4^{4k^2} + 3$ or a well-connected set $Y$ of size at least $2(4k^2) \, t \,4^{4k^2}+2$. Section~\ref{sec:bounded} deals with the first case, whereas the second case is covered in Section~\ref{sec:important}.

\subsection{Small treewidth}
\label{sec:bounded}

Suppose our algorithm has found a tree-decomposition of $G$ of width at most $3(4k^2)\, t \, 4^{4k^2} + 3$. We will use the following celebrated theorem by Courcelle~\cite{Courcelle90} to solve the {\sc Cheap Coloring Extension} problem in this case.

\begin{theorem}[\cite{Courcelle90}]
\label{thm:courcelle}
There is an algorithm that tests whether an MSO formula $\psi$ holds on a graph $G$ of treewidth $w$, in time $f(|\psi|,w)\, n$.
\end{theorem}

We remark that Theorem~\ref{thm:courcelle} holds even when the input graph $G$ is supplemented by unary relations $\alpha_1, \ldots, \alpha_p$ on vertices and edges and the MSO formula $\psi$ is allowed to use these relations~\cite{Courcelle90}.

\begin{lemma}
\label{lem:solvetwlinked}
There is an algorithm that, given an instance $(G,k,t,T_1,T_2)$ of {\sc Cheap Coloring Extension} together with a tree-decomposition of $G$ of width $w$, solves the instance in time $f(k,t,w)\, n$.
\end{lemma}

\begin{proof}
The {\sc Cheap Coloring Extension} problem can be formulated in MSO. To see this, we define two unary relations on the vertex set of $G$: $\alpha_1(v)$ is true if $v \in T_1$ and $\alpha_2(v)$ is true if $v \in T_2$. We claim that the following formula holds on $G$ if and only if $G$ has a $(T_1,T_2)$-extension of cost at most $k$.
\begin{eqnarray*}
\psi & = & \exists S_1 \subseteq V(G),S_2 \subseteq V(G) ,C \subseteq E(G),e_1 \in E(G),e_2 \in E(G),\ldots,e_k \in E(G):\\
& &C=\{e_1,e_2,\ldots e_k\} \wedge \forall v \in V(G), \alpha_1(v) \rightarrow v \in S_1 \wedge \alpha_2(v) \rightarrow v \in S_2\\
& &\wedge~ \forall uv \in E(G)~\big{(}(u \in S_1 \wedge v \in S_2) \vee (u \in S_2 \wedge v \in S_1)\\
& &\vee~ (\forall X_1 \subseteq V(G),u\notin X_1 \vee v\in X_1 \vee \exists xy \in C~ (x\in X_1 \wedge y \notin X_1))\big{)}
\end{eqnarray*}
The interpretation of $\psi$ is as follows. The sets $S_1$ and $S_2$ are the vertices of $G$ colored $1$ and $2$, respectively, by a $(T_1,T_2)$-extension $\phi$. The edge set $C$ contains all the edges, $k$ in total, of a spanning tree of each monochromatic component of $\phi$. Every edge $uv$ of $G$ is either good (third line of the formula), or bad (fourth line). If $uv$ is a bad edge and $uv$ belongs to $C$, then we can take $xy=uv$. Suppose $uv$ is a bad edge and $uv\notin C$. Let $X$ be the monochromatic component of $\phi$ such that $G[X]$ contains $uv$, and let $X_1$ be any subset of $V(G)$ containing $u$ but not $v$. Since $C$ contains all the edges of a spanning tree of $G[X]$, there exists a path $P$ in $G[X]$ from $u$ to $v$, using only edges of $C$. Hence we can take $xy$ to be the first edge of $P$ that has one endpoint ($x$) in $X_1$ and the other endpoint ($y$) in $V(G)\setminus X_1$.

Note that, for simplicity, we took some liberties in the formulation of the MSO formula. For example, $C=\{e_1,e_2,\ldots e_k\}$ is not really an MSO formulation, but it can easily be translated into MSO by demanding that every edge $e_i \in C$ and that any edge in $C$ must be one out of $e_1,e_2,\ldots e_k$. Similarly, operators such as $\rightarrow$ can be reformulated using the $\wedge$, $\vee$, and $\neg$ operators. Applying Theorem~\ref{thm:courcelle} to this formulation completes the proof of the lemma.
\qed
\end{proof}

We would like to remark that, given an instance $(G,k,t,T_1,T_2)$ of {\sc Cheap Coloring Extension} together with a tree-decomposition of $G$ of width $w$, it is possible to solve that instance in time $(w+1)^{O(w)} n$, using standard dynamic programming techniques. This gives a much faster algorithm than the one obtained by applying Theorem~\ref{thm:courcelle} on the MSO formula. However, it would take several pages to give the details of such an algorithm, and for the main purpose of this paper, we find it sufficient to handle the case of small treewidth by Lemma~\ref{lem:solvetwlinked}.

\subsection{Large treewidth and irrelevant edges}
\label{sec:important}

Suppose our algorithm did not find a tree-decomposition of $G$ of small width, but instead found a well-connected set $Y$ of size at least $2(4k^2) \, t \,4^{4k^2}+2$. We use $Y$ throughout this section to refer to this specific set. An edge $e \in E(G)$ is said to be {\em irrelevant} if it satisfies the following property: $G$ has a $(T_1,T_2)$-extension of cost at most $k$ if and only if $G-e$ does. We will show that the presence of the large well-connected set $Y$ guarantees the presence of an irrelevant edge $e$ in $G$. Hence we find such an irrelevant edge $e$ in $G$, delete it from the graph, and solve {\sc Cheap Coloring Extension} on the instance $(G-e,k,t,T_1,T_2)$. Since each iteration of this process deletes an edge, we will find a tree-decomposition of the graph under consideration of small width after at most $m$ iterations, in which case we solve the problem as described in Section~\ref{sec:bounded}. 

Observation~\ref{obs:irredge} below gives a hint about how we are going to identify an irrelevant edge of $G$. We first need a basic observation about bad edges.

\begin{observation}
\label{obs:goodbad}
Let $\phi$ be a 2-coloring of $G$. No bad edge has both endpoints in the same good component of $\phi$.
\end{observation}

\begin{proof}
Suppose, for contradiction, that $G$ has a bad edge $uv$ such that both $u$ and $v$ belong to a good component $C$ of $\phi$. Since $uv$ is bad, we have $\phi(u)=\phi(v)$. Every good component is connected, so there is a path $P$ in $C$, starting in $u$ and ending in $v$, consisting only of good edges. The path $P$ must contain an even number of edges, implying that $P$ and $uv$ together form an odd cycle in $G$. This contradicts the assumption that $G$, which is part of the instance $(G,k,t,T_1,T_2)$ of {\sc Cheap Coloring Extension} that we are solving, is bipartite.
\qed
\end{proof}

\begin{observation}
\label{obs:irredge}
Let $uv\in E(G)$. If $\phi$ is a cheapest $(T_1,T_2)$-extension of $G - {uv}$ and $u$ and $v$ are in the same good component of $\phi$, then $uv$ is irrelevant.
\end{observation}

\begin{proof}
Suppose $u$ and $v$ belong to the same good component of a cheapest $(T_1,T_2)$-extension $\phi$ of $G-uv$. Note that $\phi$ is a 2-coloring of $G$, and that the edge $uv$ in $G$ is good with respect to $\phi$ as a result of Observation~\ref{obs:goodbad}. Hence $\phi$ is a $(T_1,T_2)$-extension of $G$, and the cost of $\phi$ in $G$ equals the cost of $\phi$ in $G-uv$. As a result of Observation~\ref{obs:remedge}, $\phi$ must be cheapest $(T_1,T_2)$-extension of $G$. Since the cost of a cheapest $(T_1,T_2)$-extension of $G-uv$ equals the cost of a cheapest $(T_1,T_2)$-extension of $G$, the edge $uv$ is irrelevant by definition.
\qed
\end{proof}

In order to use Observation~\ref{obs:irredge}, we need to identify vertices that will end up in the same good component of some cheapest $(T_1,T_2)$-extension of $G - {uv}$. The vertices in $Y$ are good candidates, because they are so highly connected to each other. Over the next few lemmas we formalize this intuition. We start with two observations that will allow us, in the proof of Lemma~\ref{lem:boundisect} below, to bound the number of bad edges and the number of good components of a cheapest $(T_1,T_2)$-extension of $G$ of cost at most $k$.

\begin{observation}
\label{obs:badlimit}
Let $\phi$ be a $2$-coloring of $G$. If $\phi$ has cost at most $k$, then there are less than $2k^2$ bad edges.
\end{observation}

\begin{proof}
Let ${\cal M}_\phi' = \{X \in{\cal M}_\phi~:~|X|\geq 2\}$ be the set of monochromatic components of $\phi$ containing more than one vertex, and let $G'$ be the disjoint union of the graphs induced in $G$ by the elements of ${\cal M}'_\phi$, i.e., $G'=\bigcup_{X\in {\cal M}'_\phi} G[X]$. By definition, the cost of $\phi$ is $\sum_{X\in {\cal M}_\phi} (|X|-1)=\sum_{X\in {\cal M}'_\phi} (|X|-1)$, which is exactly the number of edges in any spanning forest of $G'$. Since any forest on at most $k$ edges without isolated vertices has at most $2k$ vertices, we have $|V(G')|\leq 2k$. Every bad edge has both endpoints in $V(G')$, so the number of bad edges is at most ${2k\choose 2} < 2k^2$.
\qed
\end{proof}

\begin{observation}
\label{obs:containterm}
Let $\phi$ be a cheapest $(T_1,T_2)$-extension of $G$. Every good component of $\phi$ contains a vertex from $T_1 \cup T_2$.
\end{observation}

\begin{proof}
Suppose a good component $C$ of $\phi$ does not contain any vertex from $T_1 \cup T_2$. We build a coloring $\phi'$ from $\phi$ by changing the color of every vertex in $C$, leaving the color of every other vertex unchanged, i.e., $\phi'(v)=3-\phi(v)$ if $v\in C$, and $\phi'(v)=\phi(v)$ if $v\notin C$. Since $\phi(v)=\phi'(v)$ for every $v\in T_1 \cup T_2$, $\phi'$ is a $(T_1,T_2)$-extension of $G$. Furthermore, every edge that was good with respect to $\phi$ is good with respect to $\phi'$, while every edge in $\delta_G(C)$ was bad with respect to $\phi$ and is good with respect to $\phi'$. Since $G$ is connected, there is some vertex $v \in C$ which is incident to at least one edge that was bad with respect to $\phi$. On the other hand, all edges incident to $v$ are good with respect to $\phi'$. Hence $\{v\}$ is a monochromatic component of $\phi'$, but $\{v\}$ was not a monochromatic component of $\phi$. This means that $|{\cal M}_\phi| < |{\cal M}_{\phi'}|$. This, together with the observation that the number of edges that are bad with respect to $\phi'$ is not more than the number of edges that were bad with respect to $\phi$, implies that the cost of $\phi'$ is strictly less than the cost of $\phi$. This contradicts the assumption that $\phi$ is a cheapest $(T_1,T_2)$-extension of $G$.
\qed
\end{proof}

The next lemma shows that almost all the vertices of $Y$ appear in the same good component of any cheapest $(T_1,T_2)$-extension $\phi$ of $G$ of cost at most $k$. 

\begin{lemma}
\label{lem:boundisect}
Let $uv\in E(G)$, and let $\phi$ be a cheapest $(T_1,T_2)$-extension of $G - uv$ of cost at most $k$. There exists exactly one good component $C^*$ of $\phi$ satisfying $|Y \setminus C^*| \leq 2k^2$, and every other good component $C'$ of $\phi$ satisfies $|Y \cap C'| \leq 2k^2$.
\end{lemma}

\begin{proof}
Let $C$ be a good component of $\phi$. We first show that either $|Y \setminus C| \leq 2k^2$ or $|Y \cap C| \leq 2k^2$. Suppose for contradiction that $|Y \cap C| > 2k^2$ and $|Y \setminus C| > 2k^2$. We define $Y_1$ to be the smallest of the two sets $Y\cap C$ and $Y\setminus C$, and $Y_2$ to be any subset of the largest of the two sets such that $|Y_2|=|Y_1|$. Note that $2k^2+1\leq |Y_1|=|Y_2|\leq |Y|/2$. By the definition of a well-connected set, there are $|Y_1| \geq 2k^2+1$ vertex-disjoint paths with one endpoint in $Y \cap C$ and the other in $Y \setminus C$. At least $2k^2$ of these paths exist in $G - uv$, and each of those must contain an edge in $\delta_{G-uv}(C)$. Since each edge in $\delta_{G-uv}(C)$ is bad, it follows that $\phi$ has at least $2k^2$ bad edges, contradicting Observation~\ref{obs:badlimit}.

Now suppose for contradiction that $\phi$ does not have a good component $C^*$ with $|Y\setminus C^*|\leq 2k^2$. Then $|Y\cap C|\leq 2k^2$ for every good component $C$ of $\phi$, as we showed earlier. Since $\phi$ has at most $t=|T_1|+|T_2|$ good components as a result of Observation~\ref{obs:containterm}, at most $t\, 2k^2$ vertices of $Y$ appear in good components. The fact that the size of $Y$ is much larger than $t\, 2k^2$, together with the observation that every vertex of $G$ appears in a good component by definition, yields the desired contradiction. Hence we know that $\phi$ has a good component $C^*$ with $|Y\setminus C^*|\leq 2k^2$. The uniqueness of $C^*$ follows from the sizes of $Y$ and $C^*$, and the fact that the good components of $\phi$ are pairwise disjoint.
\qed
\end{proof}

There are two problems with how to exploit the knowledge obtained from Lemma~\ref{lem:boundisect}. The first is that, even though we know that almost all the vertices of $Y$ appear in the same good component together, we do not know exactly which ones do. The second problem is that we are looking for an {\em edge} with both endpoints in the same good component, and $Y$ could be an independent set and thus not immediately give us an edge to delete. We deal with both problems by employing the very useful notion of {\em important sets}. For two vertices $x,y \in V(G)$, we say that a set $X\subseteq V(G)$ is $(x,y)$-{\em important} if it satisfies the following three properties: (1) $x \in X$ and $y \notin X$; (2) $G[X]$ is connected; and (3) there is no $X'\supset X$, $y\not\in X'$ such that $d_G(X')\le d_G(X)$ and $G[X']$ is connected. The following theorem was first proved in~\cite{chenll07}. We use here the formulation in~\cite{ML11}, because that one best fits the purposes of this paper.

\begin{theorem}[\cite{chenll07,ML11}]
\label{thm:impsep}
Let $x,y$ be two vertices in a graph $G$. For every $p\ge 0$, there are at most $4^p$ $(x,y)$-important sets $X$ such that $d_G(X) \leq p$. Furthermore, these important sets can be enumerated in time $4^p\cdot n^{O(1)}$.
\end{theorem}

Suppose $G - uv$ has a cheapest $(T_1,T_2)$-extension $\phi$ of cost at most $k$. We will use the important sets together with Lemma~\ref{lem:boundisect} to identify vertices in $Y$ which must be in the unique good component $C^*$ of $\phi$ that contains all but at most $2k^2$ vertices of $Y$. We first build a graph $G^*$ from $G$ by adding a new vertex $y^*$ and making $y^*$ adjacent to all vertices in $Y$. We then enumerate all $x \in T_1 \cup T_2$ and all $(x,y^*)$-important sets $X$ in $G^*$ such that $d_{G^*}(X) \leq 4k^2$. By Theorem~\ref{thm:impsep}, this can be done in time $4^{4k^2} n^{O(1)}$ for each choice of $x$. Finally, we define the set $Z$ to be the union of all enumerated sets $X$. In other words,
\begin{eqnarray*}
Z & &= \{ w \in V(G^*)~:~\\
& & \exists x\in T_1 \cup T_2, X \subseteq V(G^*), w \in X, d_{G^*}(X) \leq 4k^2 \mbox{ and } X \mbox{ is } (x,y^*)\mbox{-important}\}
\end{eqnarray*}
Observe that, given $G$ and $Y$, $Z$ can be computed in time $t \, 4^{4k^2} n^{O(1)}$. We will use the set $Z$ in the following way. First we show that if there is an edge $uv \in E(G)$ such that neither $u$ nor $v$ are in $Z$, then the edge $uv$ is irrelevant. Then we show that such an edge always exists. 

\begin{lemma}
\label{uv-irrelevant}
Let $uv \in E(G)$ such that $u \notin Z$ and $v \notin Z$. Then $uv$ is irrelevant. 
\end{lemma}

\begin{proof}
Let $\phi$ be a cheapest $(T_1,T_2)$-extension of $G - uv$ with cost at most $k$. Let $C^*$ be a good component of $\phi$ such that $|Y\setminus C^*| \leq 2k^2$. By Lemma~\ref{lem:boundisect}, such a component $C^*$ exists, and every other good component $C$ of $\phi$ satisfies $|Y \cap C| \leq 2k^2$. We prove that both $u$ and $v$ are in $C^*$. Suppose $u \notin C^*$. Then $u \in C$ for some other good component of $\phi$. Now, $C$ induces a connected subgraph in $G - uv$, and all edges leaving $C$ in $G - uv$ are bad with respect to $\phi$. Since $\phi$ has less than $2k^2$ bad edges by Observation~\ref{obs:badlimit}, it follows that $d_{G-uv}(C) < 2k^2$, and thus $d_G(C) \leq 2k^2$. Furthermore, because $|Y \cap C| \leq 2k^2$, we have that $d_{G^*}(C) \leq 4k^2$. Finally, by Observation~\ref{obs:containterm}, $C$ must contain a vertex $x \in T_1 \cup T_2$. Hence there must be a $(x,y^*)$-important set $X$ such that $C \subseteq X$ and $d(X) \leq 4k^2$ in $G^*$. But $C \subseteq X \subseteq Z$, which implies that $u \in Z$, contradicting the assumption that $u\notin Z$. The proof that $v \in C^*$ is identical. We conclude that both $u$ and $v$ are in $C^*$, and hence, by Observation~\ref{obs:irredge}, $uv$ is irrelevant. 
\qed
\end{proof}

\begin{lemma}
\label{alwaysirrelevant} 
$G$ contains an edge $uv$ such that $u \notin Z$ and $v \notin Z$. 
\end{lemma}

\begin{proof}
We first prove that $d_G(Z) \leq (4k^2) \, t \,4^{4k^2}$ and $|Z \cap Y| \leq (4k^2) \, t \, 4^{4k^2}$. For the first inequality, it suffices to show that $d_{G^*}(Z) \leq (4k^2)\, t\, 4^{4k^2}$, because $G$ is a subgraph of $G^*$. However, $Z$ is the union of at most $t \, 4^{4k^2}$ important sets $X$, where $d_{G^*}(X) \leq 4k^2$ for each set. Hence the first inequality follows. To see that $|Z \cap Y| \leq (4k^2)\, t\, 4^{4k^2}$, observe that each $(x,y^*)$-important set $X$ in $G^*$ with $d_{G^*}(X) \leq 4k^2$ contains at most $4k^2$ vertices of $Y$, since each vertex in $Y$ is a neighbour of $y^*$.

In order to prove that $G$ contains an edge $uv$ with $u\notin Z$ and $v\notin Z$, we arbitrarily partition $Y$ into $Y_1$ and $Y_2$ such that $|Y_1|=|Y_2|$ and $Z \cap Y \subseteq Y_2$. Since $|Y|\geq 2(4k^2) \, t \,4^{4k^2}+2$ by assumption and we showed that $|Z \cap Y|\leq (4k^2)\, t\, 4^{4k^2}$, such a partition always exists. By the definition of a well-connected set, there are $|Y_1|=(4k^2)\, t\, 4^{4k^2}+1$ vertex-disjoint paths starting in $Y_1$ and ending $Y_2$. For every $i \leq|Y_1|$, let $u_iv_i$ be the first edge on the $i$th such path, with $u_i \in Y_1$. Recall that $Z \cap Y \subseteq Y_2$ by assumption. Since all of the $u_i$'s are in $Y_1$, none of them are in $Z$. Thus, each $v_i$ that belongs to $Z$ contributes one to $d_G(Z)$, as then $u_iv_i\in \delta_G(Z)$. Since we bounded $d_G(Z)$ from above by $(4k^2) \, t\, 4^{4k^2}$ at the start of this proof, not every $v_i$ can belong to $Z$. Hence there is an edge $u_iv_i$ with neither endpoint in $Z$.
\qed
\end{proof}

We are now ready to state the main lemma of this section.

\begin{lemma}
\label{lem:extension}
{\sc Cheap Coloring Extension} can be solved in time $f(k,t) \, n^{O(1)}$.
\end{lemma}

\begin{proof}
Let $(G,k,t,T_1,T_2)$ be an instance of {\sc Cheap Coloring Extension}, and let $f$ be an appropriate function that does not depend on $n$. We first apply Theorem~\ref{thm:twlinked} and the remark immediately following it to compute, in time $f(k,t) \, n^{O(1)}$, either a tree-decomposition of $G$ of width at most $3(4k^2)\, t \, 4^{4k^2} + 3$ or a well-connected set $Y$ of size at least $2(4k^2) \, t \,4^{4k^2}+2$. If we get a tree-decomposition of small width, we apply Lemma~\ref{lem:solvetwlinked} to solve the problem in additional time $f(k,t) \, n$, after which we terminate and output the answer. If we find a well-connected set $Y$, we continue to find an irrelevant edge $e \in E(G)$ and delete it from $G$. Lemmas~\ref{uv-irrelevant} and~\ref{alwaysirrelevant} guarantee that such an edge always exists. In order to find $e$, we first compute the set $Z$. We already argued that this can be done in time $f(k,t) \, n^{O(1)}$. We can find an irrelevant edge as explained in the proof of Lemma~\ref{alwaysirrelevant} in additional polynomial time, since this amounts to computing $Z \cap Y$, choosing $Y_2$ to contain the whole intersection and as many more vertices as needed to obtain $|Y_1|=|Y_2|$, and checking all edges leaving $Y_1$ to find one whose endpoints do not belong to $Z$. The total running time of this whole procedure is clearly $f(k,t) \, n^{O(1)}$.

After an irrelevant edge $e$ is deleted from $G$, we run the whole procedure on $(G-e,k,t,T_1,T_2)$. This can be repeated at most $|E(G)|=n^{O(1)}$ times, and hence the total running time $f(k,t) \, n^{O(1)}$ follows.
\qed
\end{proof}

Our main result immediately follows from Lemmas~\ref{lem:contrcolor},~\ref{lem:cheaptocheaper},~\ref{lem:redextension}, and~\ref{lem:extension}.

\begin{theorem}
{\sc Bipartite Contraction} is fixed parameter tractable when parameterized by $k$.
\end{theorem}

We end this section with a remark on the running time. If we use Lemma \ref{lem:solvetwlinked} in the proof of Lemma \ref{lem:extension}, then the parameter dependence of the whole algorithm is dominated by a very large function in $k$ \cite{Courcelle90}. However, as we remarked after Lemma \ref{lem:solvetwlinked}, we can obtain a running time of $(4^{O(k^2)})^{4^{O(k^2)}} n= 2^{2^{O(k^2)}} n$ for the small treewidth case of the proof of Lemma \ref{lem:extension}. This is because the treewidth of the instance at hand is $4^{4k^2} k^{O(1)} = 4^{4k^2 O(\log k)}=4^{O(k^2)}$ when the small treewidth case applies. This gives a total running time of of $2^{2^{O(k^2)}} n^{O(1)}$ for our algorithm for {\sc Bipartite Contraction}.

\section{Concluding remarks}

For completeness, we would like to mention that {\sc Bipartite Contraction} is NP-complete; a polynomial time reduction from {\sc Edge Bipartization} can be obtained by replacing every edge of the input graph by a path of sufficiently large odd length. A highly relevant question is whether {\sc Bipartite Contraction} admits a polynomial kernel, meaning that there is a polynomial time algorithm that transforms an instance $(G,k)$ to an instance $(G',k')$ of size $g(k)$, where $g$ is a polynomial in $k$. The mentioned NP-completeness reduction from  {\sc Edge Bipartization} is parameter preserving, as $k$ remains the same. However, it is not known whether or not {\sc Edge Bipartization} admits a polynomial kernel.

\begin{footnotesize}

\end{footnotesize}

\end{document}